\documentclass[11pt,english]{article}
\usepackage[latin9]{inputenc}
\usepackage[a4paper]{geometry}
\usepackage{color}
\usepackage{amsmath}
\usepackage{amsthm}
\usepackage{amssymb}
\usepackage{amsmath}
\usepackage{mathcomp}
\usepackage{dsfont}
\usepackage[toc,page]{appendix}
\usepackage{graphicx} 
\usepackage{subfigure}
\usepackage{enumerate}
\usepackage[dvipsnames]{xcolor}
\usepackage{authblk}

\usepackage{hyperref}

\newtheorem{definition}{Definition}[section]
\newtheorem{lemma}[definition]{Lemma}

\newtheorem{proposition}[definition]{Proposition}

\newtheorem{theorem}[definition]{Theorem}
\newtheorem{remark}[definition]{Remark}

\numberwithin{equation}{section}

\def\rd{\mathrm{d}}

\def\tr{\mathrm{tr}}

\def\ad{\mathrm{ad}}

\def\N{\mathcal{N}_+}

\def\bR{\mathbb{R}} 
\def\bN{\mathbb{N}}

\def\cN{\mathcal{N}} 
\def\cU{\mathcal{U}} 
\def\cF{\mathcal{F}}
\def\cM{\mathcal{M}}
\def\cG{\mathcal{G}}
\def\cJ{\mathcal{J}}
\def\cL{\mathcal{L}}
\def\cW{\mathcal{W}} 
\def\ph{\varphi}
\def\wt{\widetilde}

\newcommand{\vertiii}[1]{{\left\vert\kern-0.25ex\left\vert\kern-0.25ex\left\vert #1 
    \right\vert\kern-0.25ex\right\vert\kern-0.25ex\right\vert}}

\title{A large deviation principle in many-body quantum dynamics}
\author[1]{Kay Kirkpatrick}
\author[2]{Simone Rademacher}
\author[3]{Benjamin Schlein} 
\affil[1]{Department of Mathematics, University of Illinois at Urbana-Champaign Urbana, IL 61801, USA \\kkirkpat@illinois.edu}
\affil[2]{Institute of Science and Technology Austria, Am Campus 1, 3400 Klosterneuburg, Austria\\ simone.rademacher@ist.ac.at 
}
\affil[3]{Institute of Mathematics, University of Zurich, Winterthurerstrasse 190, 8057 Zurich, Switzerland\\ benjamin.schlein@math.uzh.ch}

\begin{document}

\maketitle

\begin{abstract}
We consider the many-body quantum evolution of a factorized initial data, in the mean-field regime. We show that fluctuations around the limiting Hartree dynamics satisfy large deviation estimates, that are consistent with central limit theorems that have been established in the last years.    
\end{abstract}


\section{Introduction}

A system of $N$ bosons in the mean-field regime can be described by the Hamilton operator 
\[ H_N = \sum_{j=1}^N -\Delta_{x_j} + \frac{1}{N} \sum_{i < j}^N v (x_i - x_j) \]
acting on the Hilbert space $L^2_s (\bR^{3N})$, the subspace of $L^2 (\bR^{3N})$ consisting 
of functions that are symmetric with respect to any permutation of the $N$ particles. 

The time evolution of the $N$ particles is governed by the many-body Schr\"odinger equation 
\begin{equation}\label{eq:schrN}
i\partial_t \psi_{N,t} = H_N \psi_{N,t} \; .
\end{equation}

If the $N$ particles are trapped into a finite region by a confining external potential $v_\text{ext}$, the system exhibits, at zero temperature, complete Bose-Einstein condensation in the minimizer of the Hartree energy functional 
\[ \mathcal{E}_\text{Hartree} (\ph) = \int \left[ |\nabla \ph|^2 + v_\text{ext} |\ph|^2 \right] dx + \frac{1}{2} \int v (x-y) |\ph (x)|^2 |\ph (y)|^2 dx dy \]
taken over $\ph \in L^2 (\bR^3)$ with $\| \ph \| = 1$. For this reason, from the point of view of physics, it is interesting to study the solution of (\ref{eq:schrN}) for an initial sequence $\psi_N \in L^2_s (\bR^{3N})$ exhibiting complete Bose-Einstein condensation, in the sense that the one-particle reduced density $\gamma_N = \tr_{2, \dots , N} |\psi_N \rangle \langle \psi_N |$ associated with $\psi_N$ satisfies $\gamma_N \to |\ph \rangle \langle \ph|$ for a normalized one-particle orbital $\ph \in L^2 (\bR^3)$, in the limit $N \to \infty$. 

To keep our analysis as simple as possible, we consider solutions of (\ref{eq:schrN}) for factorized initial data $\psi_{N,0} = \ph^{\otimes N}$ (which obviously exhibits condensation, since $\gamma_N = |\ph \rangle \langle \ph|$). Notice, however, that our approach could be extended to physically more interesting initial data exhibiting condensation. 

Under quite general assumptions on the interaction potential $v$, one can show that (in contrast with factorization) the property of Bose-Einstein condensation is preserved by the many-body evolution (\ref{eq:schrN}) and that, for every fixed $t \in \bR$, the reduced one-particle density $\gamma_{N,t} = \tr_{2, \dots , N} |\psi_{N,t} \rangle \langle \psi_{N,t}|$ is such that $\gamma_{N,t} \to |\ph_t \rangle \langle \ph_t|$, as $N \to \infty$. Here $\ph_t$ is the solution of the nonlinear Hartree equation 
\begin{equation}\label{eq:hartree} i \partial_t \ph_t = -\Delta \ph_t + (v * |\ph_t|^2) \ph_t \end{equation}
with the initial data $\ph_{t=0} = \ph$. See for example \cite{AGT,AFP,BGM,CLS,C,EY,FKP,FKS,GV,KP,RS,Sp}.

The convergence $\gamma_{N,t} \to |\ph_t \rangle \langle \ph_t|$ of the reduced one-particle density associated with the solution of the Schr\"odinger equation (\ref{eq:schrN}) can be interpreted as a law of large numbers. For a self-adjoint operator $O$ on $L^2 (\bR^3)$, let $O^{(j)} = 1 \otimes \dots \otimes O \otimes \dots \otimes 1$ denote the operator on $L^2 (\bR^{3N})$ acting as $O$ on the $j$-th particle and as the identity on the other $(N-1)$ particles. The probability that, in the state described by the wave function $\psi \in L^2_s (\bR^{3N})$, the observable $O^{(j)}$ takes values in a set $A \subset \bR$ is determined by 
\[ \mathbb{P}_\psi (O^{(j)} \in A) = \big\langle \psi, \chi_A (O^{(j)}) \psi \big\rangle \, . \]
For factorized wave functions $\psi_N = \ph^{\otimes N}$, the operators $O^{(j)}$, $j=1, \dots , N$, define independent and identically distributed random variables with average $\langle \ph , O \ph \rangle$. The standard law of large numbers implies that 
\[ \lim_{N \to \infty} \mathbb{P}_{\ph^{\otimes N}} \left( \left| \frac{1}{N} \sum_{j=1}^N O^{(j)} - \langle \ph , O \ph \rangle \right| > \delta \right) = 0  \]
for all $\delta > 0$. The solution $\psi_{N,t}$ of the Schr\"odinger equation (\ref{eq:schrN}), with factorized initial data $\psi_{N,0} = \ph^{\otimes N}$, is not factorized. Nevertheless, the convergence of the reduced density $\gamma_{N,t} \to |\ph_t \rangle \langle \ph_t|$ implies that the law of large numbers still holds true, i.e. that 
\begin{equation}\label{eq:LLN} \lim_{N \to \infty} \mathbb{P}_{\psi_{N,t}} \left( \left| \frac{1}{N} \sum_{j=1}^N O^{(j)} - \langle \ph_t , O \ph_t \rangle \right| > \delta \right) = 0  \end{equation}
for all $\delta > 0$; see, for example, \cite{BKS}. 

To go beyond (\ref{eq:LLN}) and study fluctuations around the limiting Hartree dynamics, it is useful to factor out the condensate. 

To reach this goal, we define the bosonic Fock space $\cF =  \bigoplus_{j=0}^N L^2_{\perp \ph_t} (\bR^3)^{\otimes_s j}$. On $\cF$, for any $f \in L^2 (\bR^3)$, we introduce the usual creation and annihilation operators $a^*(f), a(f)$, satisfying canonical commutation relations. It will also be convenient to use operator-valued distributions $a_x^*, a_x$, for $x \in \bR^3$, so that 
\[ a^* (f) = \int f(x) \, a_x^* \, dx , \qquad a(f) = \int \bar{f} (x) \, a_x \, dx \]
In terms of $a_x^*, a_x$, we can express the number of particles operator, defined by $(\cN \Psi)^{(n)} = n \Psi^{(n)}$, as 
\[ \cN = \int dx \, a_x^* a_x \]
More generally, for an operator $A$ on the one-particle space $L^2 (\bR^3)$, its second quantization $d\Gamma (A)$, defined on $\cF$ so that $(d\Gamma (A) \Psi)^{(n)} = \sum_{j=1}^n A_j \Psi^{(n)}$, with $A_j = 1 \otimes \dots \otimes A \otimes \dots \otimes 1$ acting non-trivially on the $j$-th particle only, can be written as 
\[ d\Gamma (A) = \int dx dy \; A(x;y) \, a_x^* a_y \]
where $A(x;y)$ is the integral kernel of $A$ (with this notation $\cN = d\Gamma (1)$). More details on the formalism of second quantization applied to the dynamics of mean-field systems can be found in \cite{BPS}.   

In order to factor out the condensate, described at time $t \in \bR$, by the solution $\ph_t$ of (\ref{eq:hartree}), we observe now that every $\psi \in L^2_s (\bR^{3N})$ can be uniquely written as 
\[ \psi = \eta_0 \ph_t^{\otimes N} + \eta_1 \otimes_s \ph^{\otimes (N-1)}_t + \dots + \eta_N \]
with $\eta_j \in L^2_{\perp \ph_t} (\bR^3)^{\otimes_s j}$, where $L^2_{\perp \ph_t} (\bR^3)$ denotes the orthogonal complement in $L^2 (\bR^3)$ of the condensate wave function $\ph_t$. This remark allows us to define, for every $t \in \bR$, a unitary operator \[ \cU_t : L^2_s (\bR^{3N}) \to \cF_{\perp \ph_t}^{\leq N} = \bigoplus_{j=0}^N L^2_{\perp \ph_t} (\bR^3)^{\otimes_s j} \]
by setting $\cU_t \psi = \{ \eta_0, \eta_1 , \dots , \eta_N \}$. The unitary map $\cU_t$, first introduced in \cite{LNSS}, removes the condensate wave function $\ph_t$ and allows us to focus on its orthogonal excitations. It maps the $N$-particle space $L^2_s (\bR^{3N})$ into the truncated Fock space $\cF_{\perp \ph_t}^{\leq N}$, constructed over the orthogonal complement of $\ph_t$. 

The map $\cU_t$ can be used to define the fluctuation dynamics (mapping the orthogonal excitations of the condensate at time $t_1$ into the orthogonal excitations of the condensate at time $t_2$): 
\begin{equation}\label{eq:cWNt} \cW_N (t_2 ; t_1) = \cU_{t_2} e^{-i H_N (t_2-t_1)} \cU_{t_1}^* : \cF_{\perp \ph_{t_1}}^{\leq N} \to \cF_{\perp \ph_{t_2}}^{\leq N} \;.  \end{equation} 
The fluctuation dynamics satisfies the equation 
\[ i\partial_{t_2} \cW_N (t_2 ; t_1) = \cL_N (t_2) \cW_N (t_2 ; t_1) \]
with $\cW_N (t_1;t_1) = 1$ for all $t_1 \in \bR$ and with the generator $\cL_N (t) = \left[ i\partial_{t} \cU_{t} \right] \cU_{t}^* + \cU_{t} H_N \cU_{t}^*$.  To compute the generator $\cL_N (t)$, we use the rules 
\begin{equation}\label{eq:rules}
\begin{split} 
\cU_{t} a^* (\ph_{t}) a (\ph_{t}) \cU_{t}^* &= N - \cN_+ ({t}) , \\
\cU_{t} a^* (f) a (\ph_{t}) \cU_{t}^* &= a^* (f) \, \sqrt{N - \cN_+ ({t})}  , \\
\cU_{t} a^* (\ph_{t}) a (f) \cU_{t}^* &= \sqrt{N - \cN_+ ({t})} \, a(f)  ,\\
\cU_{t} a^* (f) a (g) \cU_{t}^* &= a^* (f) a(g) 
\end{split} \end{equation}
for any $f,g \in L^2_{\perp \ph_{t}} (\bR^3)$. We obtain, similarly to \cite{LNS}, the matrix elements  
\begin{equation}\label{eq:cLN} \begin{split} 
\langle \xi_1, \cL_N (t) \xi_2 \rangle = \; & \langle \xi_1, d\Gamma (h_H (t) + K_{1,t})  \xi_2 \rangle + \text{Re} \int dxdy \;  K_{2,t} (x;y) \, \langle \xi_1, b_x^* b_y^* \xi_2 \rangle \\ 
&- \frac{1}{2N} \langle \xi_1 , d\Gamma (v *|\ph_{t}|^2 + K_{1,t} - \mu_{t}) (\cN_+ (t) - 1) \xi_2 \rangle \\ &+ \frac{2}{\sqrt{N}} \text{Re} \, \langle \xi_1, \cN_+  b ((v*|\ph_{t}|^2) \ph_{t}) \xi_2 \rangle \\ &+\frac{2}{\sqrt{N}} \int dx dy \;  v (x-y)  \text{Re} \,  \ph_{t} (x) \langle \xi_1, a_y^*a_{x'}b_{y'} \xi_2 \rangle   \\ &+ \frac{1}{2N} \int dx dy \,v (x-y) \langle \xi_1 , a_x^* a_y^* a_{x} a_{y} \xi_2 \rangle \, .  \end{split} \end{equation} 
for any $\xi_1, \xi_2 \in \cF_{\perp \ph_{t}}^{\leq N}$. Here $h_H (t) = -\Delta + (v * |\ph_{t}|^2)$, $K_{1,t} (x;y) = v(x-y) \ph_{t} (x) \overline{\ph}_{t} (y)$, $K_{2,t} (x;y) = v(x-y) \ph_{t} (x) \ph_{t} (y)$, $2\mu_{t} = \int dx dy \; v(x-y)  \vert \varphi_{t} (x) \vert^2  \vert \varphi_{t} (y) \vert^2 $. Moreover, we introduced the notation $\cN_+ (t)$ for the number of particles operator on the space $\cF_{\perp \ph_{t}}^{\leq N}$ ($\cN_+ (t) = d\Gamma (q_{t})$, with $q_{t} = 1 - |\ph_{t} \rangle \langle \ph_{t}|$, if we think of $\cF_{\perp \ph_{t}}^{\leq N}$ as a subspace of $\cF$) and, for $f \in L^2_{\perp \ph_{t}} (\bR^3)$, we defined (using the notation introduced in \cite{BS}) 
\begin{equation}\label{eq:bbstar} \begin{split} b^* (f) &= \cU_{t} \, a^* (f) \frac{a(\ph_{t})}{\sqrt{N}} \, \cU_{t}^* = a^* (f) \sqrt{1-\frac{\cN_+ (t)}{N}}  ,\\ 
b(f) &= \cU_{t}  \frac{a^* (\ph_{t})}{\sqrt{N}} a(f) \cU_{t}^* = \sqrt{1-\frac{\cN_+ (t)}{N}} a(f) \end{split} \end{equation}
and the corresponding operator valued distributions $b^*_x, b_x$, for $x \in \bR^3$.

In the limit of large $N$, the fluctuation dynamics $\cW_N (t_2; t_1)$ can be approximated by a limiting dynamics $\cW_\infty (t_2 ;t_1) : \cF_{\perp \ph_{t_1}} =  \bigoplus_{j=0}^\infty L^2_{\perp \ph_{t_1}} (\bR^3)^{\otimes_s j}  \to \cF_{\perp \ph_{t_2}} =  \bigoplus_{j=0}^\infty L^2_{\perp \ph_{t_2}} (\bR^3)^{\otimes_s j}$ satisfying the equation 
\begin{equation}\label{eq:Winfty} i\partial_t \cW_\infty (t_2 ; t_1 ) = \cL_\infty (t_2) \cW_\infty (t_2;t_1) \end{equation} 
with the generator $\cL_\infty (t_2)$, whose matrix elements are given by 
\[ \langle \xi_1 , \cL_\infty (t_2) \xi_2 \rangle  = \langle \xi_1, d\Gamma (h_H (t_2) + K_{1,t_2}) \xi_2 \rangle  + \frac{1}{2} \int \left[ K_{2,t_2} (x;y) \langle \xi_1, a_x^* a_y^* \xi_2 \rangle  + \overline{K}_{2,t_2} (x;y) \langle \xi_1, a_x a_y \xi_2 \rangle  \right] \]
for all $\xi_1, \xi_2 \in \cF_{\perp \ph_{t_2}}$; see \cite{LNS} (this line of research started in \cite{H} and was further explored in \cite{C,GMM,MPP}; recently, an expansion of the many-body dynamics in powers of $N^{-1}$ was obtained in \cite{BPPS}). Notice that $\cL_\infty (t_2)$ acts on (a dense subspace of) the Fock space $\cF_{\perp \ph_{t_2}}$, constructed on the orthogonal complement of $\ph_{t_2}$, with no restriction on the number of particles. We have the inclusions $\cF_{\perp \ph_{t_2}}^{\leq N} \subset \cF_{\perp \ph_{t_2}} \subset \cF = \bigoplus_{j=0}^\infty L^2 (\bR^3)^{\otimes_s j}$.  Observe also that $\cL_\infty (t_2)$ is quadratic in creation and annihilation operators. It follows that the limiting dynamics $\cW_\infty (t_2;t_1)$ acts as a time-dependent family of Bogoliubov transformations (in a slightly different setting, this was shown in \cite{BKS}).
In other words, introducing the notation $A (f;g)= a (f) + a^* (\overline{g})$ for $f \in L^2_{\perp \ph_{t_2}} (\bR^3)$ and $g \in J L^2_{\perp \ph_{t_2}} (\bR^3)$, with $J$ the antilinear operator $Jf = \overline{f}$, we find 
\begin{equation}\label{eq:Winfty2} \cW^*_\infty (t_2;t_1) A (f;g) \cW_\infty (t_2;t_1) = A (\Theta (t_2;t_1) (f;g)) \end{equation} 
for a two-parameter family of operators $\Theta (t_2; t_1) : L_{\perp \ph_{t_1}}^2 (\bR^3) \oplus J L_{\perp \ph_{t_1}}^2 (\bR^3) \to L^2_{\perp \ph_{t_2}} (\bR^3) \oplus J L^2_{\perp \ph_{t_2}} (\bR^3)$. 

The convergence towards the limiting Bogoliubov dynamics (\ref{eq:Winfty}) has been used in \cite{BKS,BSS} to prove that, beyond the law of large numbers (\ref{eq:LLN}), the variables 
$O^{(j)}$ also satisfy the central limit theorem   
\begin{equation}\label{eq:CLT} 
\lim_{N\to\infty}  \mathbb{P}_{\psi_{N,t}} \left( \frac{1}{\sqrt{N}} \sum_{j=1}^N \left( O^{(j)} - \langle \ph_t , O \ph_t \rangle \right) < x \right) = \frac{1}{\sqrt{2\pi}\, \alpha_t} \int_{-\infty}^x e^{- r^2 / (2\alpha_t^2)} dr  \end{equation}
with $\alpha_t = \| f_{0;t} \|_2$. Here $f_{s;t} \in L^2_{\perp \ph_s} (\bR^3 )$ satisfies the equation (for all $0\leq s \leq t$)  
\begin{equation}\label{eq:fst}  i\partial_s f_{s;t} = (h_H (s) + K_{1,s} + J K_{2,s}) f_{s;t}, 
\end{equation} 
with $f_{t;t}=  q_t O \ph_t = O \ph_t   - \langle \varphi_t , O \varphi_t \rangle \varphi_t$, $h_H (s) = -\Delta + (v * |\ph_{s}|^2)$, $K_{1,s} (x;y) = v(x-y) \ph_{s} (x) \overline{\ph}_{s} (y)$ and  $K_{2,s }(x;y) = v(x-y) \ph_{s} (x) \ph_{s} (y)$ (the solution of (\ref{eq:fst}) is related with the family of Bogoliubov transformations $\Theta (t_1 ; t_2)$, since $\Theta (0;t) (f_{t; t} ;J f_{t;t} ) = (f_{0;t} ; J f_{0;t})$).

For singular interaction potentials, scaling as $N^{3\beta} v (N^\beta x)$ for a $0 < \beta < 1$ and converging therefore to a $\delta$-function as $N \to \infty$, the validity of a central limit theorem of the form (\ref{eq:CLT}) was recently established in \cite{R}; in this case, the correlation structure produced by the interaction affects the variance of the limiting Gaussian distribution. For $\beta =1$ (the Gross-Pitaevskii regime), the validity of a central limit theorem for the ground state was established instead in \cite{RaS}. 


%

In our main theorem, we show, for bounded interactions, a large deviation principle for the fluctuations of the many-body quantum evolution around the limiting Hartree dynamics. 
 \begin{theorem}\label{thm:main} 
 Let $v \in L^1 (\bR^3) \cap L^\infty (\bR^3)$. Let $O$ be a bounded self-adjoint operator 
 on $L^2 (\bR^3)$, with $\| \Delta O (1-\Delta)^{-1} \|_\text{op} < \infty$. Let $\ph \in H^4 (\bR^3)$, with $\| \ph \| = 1$. For $t \in \bR$, let $\psi_{N,t}$ denote the solution of the many-body Schr\"odinger equation (\ref{eq:schrN}), with initial data $\psi_{N,0} = \ph^{\otimes N}$. Then there exists a constant $C > 0$ (depending only on $\| \ph \|_{H^4}$) such that, denoting by $O^{(j)} = 1 \otimes \cdots \otimes O \otimes \cdots \otimes 1$ the operator $O$ acting only on the $j$-th particle,  
\begin{equation}\label{eq:main} \frac{1}{N} \log \mathbb{E}_{\psi_{N,t}} \, e^{\lambda \left[\sum_{j=1}^N (O^{(j)} - \langle \ph_t , O \ph_t \rangle ) \right]}  \leq \frac{\lambda^2}{2}  \alpha_t^2  + C \lambda^3 \vertiii{O}^3 \exp (C (1+ \| v \|_1 + \| v \|_\infty) \vert t\vert)  \end{equation} 
for all $\lambda \leq  \vertiii{O}^{-1} e^{-C ( \| v \|_\infty + \| v \|_1) t }$. Here we defined 
\begin{equation}\label{eq:Otriple}
\vertiii{O}  = \| \Delta O (1-\Delta)^{-1} \|_\text{op} +  (1+ \| v \|_\infty + \| v \|_1) \| O \|_\text{op}  
\end{equation} 
and $\alpha_t^2 = \| f_{0;t} \|_2^2$, with $f$ as defined in (\ref{eq:fst}). 
\end{theorem}

{\it Remark:} The result and its proof can be trivially extended to particles moving in $d$ dimensions, for any $d \in \bN \backslash \{ 0 \}$.

It follows from (\ref{eq:main}) that 
\[ \begin{split} \mathbb{P}_{\psi_{N,t}} \Big( N^{-1} \sum_{j=1}^N (O^{(j)} - \langle \ph_t , O \ph_t \rangle ) > x \Big) &=  \mathbb{P}_{\psi_{N,t}} \left( e^{-\lambda N x} \; e^{\lambda \left[ \sum_{j=1}^N \left( O^{(j)} - \langle \ph_t , O \ph_t \rangle \right) \right]}  > 1 \right) \\ 
&\leq e^{-\lambda N x} \, \mathbb{E}_{\psi_{N,t}} \, e^{\lambda \left[\sum_{j=1}^N (O^{(j)} - \langle \ph_t , O \ph_t \rangle ) \right]} \end{split} \] 
for all $0 \leq \lambda \leq \vertiii{O}^{-1} e^{-C ( \| v \|_\infty + \| v \|_1) t}$. Thus, 
\begin{equation}\label{eq:cor}
 \mathbb{P}_{\psi_{N,t}} \Big( N^{-1} \sum_{j=1}^N (O^{(j)} - \langle \ph_t , O \ph_t \rangle ) > x \Big) 
\leq e^{N \gamma (x)}  \end{equation} 
with rate function 
\[ \gamma (x) = \inf_\lambda  \; \left[ - \lambda x + \frac{\lambda^2}{2} \alpha_t^2 + C \lambda^3 \vertiii{O}^3 \exp (C (1+ \| v \|_1 + \| v \|_\infty) t) \right] \]
where the infimum is taken over all $0 \leq \lambda \leq  \vertiii{O}^{-1} \exp (-C ( \| v \|_\infty + \| v \|_1) t)$. For any fixed $t > 0$, the infimum is attained at 
\[ \lambda_x =  \frac{2x}{\alpha_t^2 + \sqrt{\alpha_t^4 + 12 C x \vertiii{O}^3 \exp (C (1+ \| v \|_1 + \| v \|_\infty) t)}} \]
if $x > 0$ is small enough (so that $\lambda_x \leq \vertiii{O}^{-1} \exp (-C ( \| v \|_\infty + \| v \|_1) t)$). This leads (again for $x > 0$ so small that $\lambda_x \leq \vertiii{O}^{-1} \exp (-C ( \| v \|_\infty + \| v \|_1) t)$) to  
\[ \begin{split}  \gamma (x) = \; &- \frac{2 x^2 \sqrt{\alpha_t^4 + 12 C x \vertiii{O}^3 \exp (C (1+ \| v \|_1 + \| v \|_\infty) t)}}{\left[ \alpha_t^2 + \sqrt{\alpha_t^4+ 12 C x \vertiii{O}^3 \exp (C (1+ \| v \|_1 + \| v \|_\infty) t)} \right]^2} \\ &+  
\frac{8C x^3 \vertiii{O}^3 \exp (C (1+ \| v \|_1 + \| v \|_\infty) t)}{\left[ \alpha_t^2 + \sqrt{\alpha_t^4+ 12 C x \vertiii{O}^3 \exp (C (1+ \| v \|_1 + \| v \|_\infty) t)} \right]^3}  \; .\end{split}  \]

Notice that, in the regime $x = y / \sqrt{N}$,  $N \gamma (x) \simeq -x^2/(2\alpha_t^2)$, which is consistent with the central limit theorem (\ref{eq:CLT}), obtained in \cite{BKS,BSS}. This shows, in particular, that the quadratic term on the r.h.s. of (\ref{eq:main}) is optimal.

To prove Theorem \ref{thm:main}, we first write the expectation on the l.h.s. of \eqref{eq:main} as 
\begin{equation} \label{eq:idea1} \begin{split}  \mathbb{E}_{\psi_{N,t}} \, e^{\lambda \left[\sum_{j=1}^N (O^{(j)} - \langle \ph_t , O \ph_t \rangle )\right] } &= \left\langle \psi_{N,t} , e^{\lambda \left[\sum_{j=1}^N (O^{(j)} - \langle \ph_t , O \ph_t \rangle ) \right] }  \psi_{N,t} \right\rangle \\  &=  \left\langle \Omega, \cW^*_N (t;0) e^{\lambda d\Gamma (q_t \wt{O}_t q_t) + \lambda \sqrt{N} \phi_+ (q_t O \ph_t )} \cW_N (t;0) \Omega \right\rangle \; . \end{split} \end{equation} 
in terms of the fluctuation dynamics introduced in (\ref{eq:cWNt}). Here we used the choice of the initial data to write \[ \psi_{N,t} = e^{-i H_N t} \ph^{\otimes N} = e^{-i H_N t} \cU_{0}^* \Omega = \cU_{t}^* \cW_N (t;0) \Omega \,  .\]  Then we applied (\ref{eq:rules}) to conjugate $\exp (\lambda [ \sum_{j=1}^N (O^{(j)} - \langle \ph_t, O \ph_t \rangle ] )$ with $\cU_t$. We introduced the notation $\wt{O}_t = O - \langle \varphi_t, O \varphi_t \rangle$. 

In the next step, motivated by the bound $\pm d\Gamma (q_t \wt{O}_t q_t) \leq c \, \| O \| \cN_+ (t)$, we control the r.h.s. of (\ref{eq:idea1}), by the product
\[ \left\langle \Omega, \cW^*_N (t;0) e^{\lambda \sqrt{N} \phi_+ (q_t O \ph_t)/2} e^{c \lambda \| O \| \mathcal{N}_+ (t)} e^{\lambda \sqrt{N} \phi_+ (q_t O \ph_t)/2} \cW_N (t;0) \Omega \right\rangle , \]
up to the exponential of a cubic expression in $\lambda$, contributing only to the last term on the r.h.s. of (\ref{eq:main}); this is the content of Lemma \ref{eq:step1-1}. In the next step, Lemma \ref{lm:step2}, we replace the fluctuation dynamics $\cW_N (t;0)$ by its limit $ \cW_\infty (t;0)$, as defined in (\ref{eq:Winfty}); as in the first step, also this replacement only produces an error cubic in $\lambda$ in (\ref{eq:main}). Describing the action of $\cW_\infty$ through the solution of (\ref{eq:fst}), we arrive at the product
\begin{equation}\label{eq:idea3} \left\langle \Omega, e^{\lambda \sqrt{N} \phi_+ (f_{0;t})/2}  e^{\lambda \kappa_t \cN_+ (0)} e^{\lambda \sqrt{N} \phi_+ (f_{0;t})/2}  \Omega \right\rangle  \end{equation} 

In the final step, Lemma \ref{lm:step3}, we estimate (\ref{eq:idea3}), concluding the proof of (\ref{eq:main}). This step makes use of the choice of product initial data (which implies that the expectation is taken in the vacuum); at the expenses of a longer proof, we could have proven Theorem \ref{thm:main} to a larger and physically more interesting class of initial data.

\section{Preliminaries} 
 
To begin with, we introduce some notation and we recall some basic facts. For a given normalized $\ph \in L^2 (\bR^3)$, we consider the Hilbert space $\cF_{\perp \ph}^{\leq N} = \bigoplus_{j=0}^N L^2_{\perp \ph} (\bR^3)^{\otimes_s j}$, with the number of particles operator $\cN_+ = d\Gamma (1-|\ph \rangle \langle \ph|)$. On $\cF_{\perp \ph}^{\leq N}$, we define the operators $b(f), b^* (f)$ as in (\ref{eq:bbstar}). We also define
\[  \phi_+ (f) = b (f) + b^* (f) , \quad \phi_- (f) = -i (b(f) - b^* (f)) \; . \]
For $g_1, g_2, g,h \in L^2_{\perp \ph} (\bR^3)$, we find the commutation relations 
\begin{equation}\label{eq:comm1} [ b(g), b(h) ]= [ b^* (g), b^* (h)] = 0, \quad  [b(g), b^* (h) ] = \langle g,h \rangle \left( 1 - \frac{\N}{N} \right) - \frac{1}{N} a^* (h) a (g)  , \end{equation}
\begin{equation}\label{eq:comm4} [ \phi_+ (h) , i \phi_- (g) ] = - 2 \text{Re } \langle h, g \rangle \left( 1 - \frac{\N}{N} \right) + \frac{1}{N} a^* (g) a (h) + \frac{1}{N} a^* (h) a (g) , \end{equation}
\begin{equation}\label{eq:comm2} [ b(h) , a^* (g_1) a (g_2) ] = \langle h, g_1 \rangle b (g_2) , \qquad [ b^* (h) , a^* (g_1) a (g_2) ] = - \langle g_2 , h \rangle b^* (g_1) , \end{equation}
\begin{equation}\label{eq:comm3} [ \phi_+ (h) , \cN_+ ] = i \phi_- (h), \quad [ i\phi_- (h), \cN_+ ] = \phi_+ (h)  .\end{equation}
More generally, 
\begin{equation}\label{eq:comm3b} 
[ \phi_+ (h) , \rd \Gamma (H) ] = i \phi_- (Hh), \quad [ i\phi_- (h), \rd \Gamma (H) ] = \phi_+ (Hh) \end{equation}
for any self-adjoint operators $H$. 

We also recall the bounds
\begin{equation}\label{eq:bd-bb*} \| b(h) \xi \| \leq \| h \|_2 \| \cN_+^{1/2} \xi \| , \qquad \| b^* (h) \xi \| \leq \| h \|_2 \| (\cN_+ + 1)^{1/2} \xi \|, \end{equation}
valid for any $h \in L^2_{\perp \ph} (\bR^3)$ and the estimate 
\begin{equation}\label{eq:bd-dG} \pm d\Gamma (H) \leq \| H \|_\text{op} \, \cN_+ \end{equation} 
for every bounded operator $H$ on $L^2_{\perp \ph} (\bR^3)$. For more details, we refer to \cite[Section 2]{BS}. 

Furthermore, we introduce the notation $\ad_{B}^{(n)} \left( A \right) $ defined for two operators $A,B$ recursively by 
\[ 
 \ad^{(0)}_B \left( A \right) = A, \quad \ad_{B}^{(n)} \left( A \right) = \left[ A,   \;  \ad_{B}^{(n-1)} \left( A \right)\right]. 
\]

\begin{lemma}
\label{lm:comm_b}
Let $h,g \in L^2_{\perp \varphi} \left( \mathbb{R}^3 \right)$.  Then 
\begin{equation}\label{eq:lm1-1} \begin{split} 
 \ad_{\sqrt{N}\phi_+(h)}^{(2n+1)}\left( b(g) \right) =& - 2^{2n} \sqrt{N} \| h \|^{2n}_2 \langle g, h \rangle \left( 1- \frac{\N}{N} \right)  \\
&+(2^{2n} -1) \frac{1}{ \sqrt{N}} \| h \|_2^{2n-2} \langle g , h \rangle a^* (h) a(h) + \frac{1}{\sqrt{N}} \| h \|_2^{2n} a^* (h) a (g) \end{split} \end{equation}
for all $n \geq 0$ and 
\begin{equation}\label{eq:lm1-2} \ad_{\sqrt{N}\phi_+(h)}^{(2n)}\left( b(g) \right) = \left( 2^{2n-1}-1\right) \| h \|^{2n-2}_2 \langle g, h \rangle \; i \phi_-(h) + \| h \|^{2n}_2 b(g) - \| h \|^{2n-2}_2 \langle g, h \rangle \; b^* (h) \end{equation}
for all $n \geq 1$. 
\end{lemma} 
 
\begin{proof}
We prove the Lemma by induction. From (\ref{eq:comm1}), we find 
\begin{align*}
\ad_{\sqrt{N} \phi_+(h)} \left(b(g) \right) = [ \sqrt{N} \phi_+ (h) , b(g) ] = \sqrt{N} [b^* (h), b(g)] = - \sqrt{N} \langle g,h \rangle \left( 1- \frac{\N}{N} \right) + \frac{1}{\sqrt{N}} a^*(h) a(g) ,
\end{align*}
in agreement with (\ref{eq:lm1-1}) (for $n=0$). Now, we assume that, for a given $n \in \mathbb{N}$, (\ref{eq:lm1-1}) holds true, and we prove (\ref{eq:lm1-2}), with $n$ replaced by $(n+1)$. To this end, we compute (using (\ref{eq:lm1-1})) 
\[ \begin{split} 
\ad^{(2n+2)}_{\sqrt{N} \phi_+(h)} \left(b(g) \right) = \; & [ \sqrt{N} \phi_+ (h), \ad^{(2n+1)}_{\sqrt{N} \phi_+(h)} \left( b(g) \right) ] \\ = \; &  2^{2n}  \| h \|_2^{2n} 
\langle g, h \rangle [ \phi_+ (h) ,  \N  ] + (2^{2n} - 1)  \| h \|_2^{2n-2} \langle g ,h \rangle [ \phi_+ (h), a^* (h) a(h) ] \\ &+ \| h \|_2^{2n}  [ \phi_+ (h) , a^* (h) a (g) ] \; .\end{split}   \]
With (\ref{eq:comm2}) and (\ref{eq:comm3}), we obtain (using the identity $2^{2n} + (2^{2n} - 1) = 2^{2n+1} -1$) 
\[ \ad^{(2n+2)}_{\sqrt{N} \phi_+(h)} \left(b(g) \right) = (2^{2n+1} -1) \| h \|^{2n}_2 \langle g, h \rangle i \phi_- (h) + \| h \|_2^{2n+2} b(g) - \| h \|_2^{2n} \langle g, h \rangle b^* (h)  \]
as claimed in (\ref{eq:lm1-2}) (with $n$ replaced by $n+1$). Finally, we assume (\ref{eq:lm1-2}) for a given $n \in \mathbb{N}$, and we show that (\ref{eq:lm1-1}) holds true, with the same $n \in \mathbb{N}$. In fact, using (\ref{eq:lm1-2}), we get  
\[ \begin{split} 
\ad^{(2n+1)}_{\sqrt{N} \phi_+(h)} \left(b(g) \right) = \; &[ \sqrt{N} \phi_+ (h), \ad^{(2n)}_{\sqrt{N} \phi_+(h)} \left( b(g) \right) ] \\ = \; &(2^{2n-1} - 1) \| h \|_2^{2n-2} \langle g ,h \rangle \sqrt{N} [ \phi_+ (h) , i \phi_- (h) ] + \| h \|_2^{2n} \sqrt{N} [ \phi_+ (h) , b(g) ] \\ &- \| h \|^{2n-2}_2 \langle g,h \rangle \sqrt{N} [ \phi_+ (h) , b^* (h) ] \; . \end{split} \]
With (\ref{eq:comm1}), (\ref{eq:comm4}), we find (using the identities $-2(2^{2n-1} - 1) - 2 = - 2^{2n}$ and $2 (2^{2n-1} - 1) + 1 = 2^{2n} - 1$), 
\[  \begin{split}  \ad^{(2n+1)}_{\sqrt{N} \phi_+(h)} \left(b(g) \right) =\; & - 2^{2n} \sqrt{N} \| h \|_2^{2n} \langle g ,h \rangle \left( 1 - \frac{\N}{N} \right) + (2^{2n} -1)  \frac{1}{\sqrt{N}}  \| h \|_2^{2n-2} \langle g ,h \rangle a^* (h) a(h) \\ &+  \frac{1}{\sqrt{N}}  \| h \|_2^{2n} a^* (h) a (g) \end{split} \]
confirming (\ref{eq:lm1-1}).
\end{proof}

\begin{proposition}\label{prop:b-con}
Let $g, h \in L^2_{\perp \varphi} (\mathbb{R}^3)$. With the shorthand notation $\gamma_s = \cosh s$ and $\sigma_s = \sinh s$, we have 
\begin{equation}
\label{eq:prop1-1} 
\begin{split} 
e^{\sqrt{N} \phi_+ (h)} b(g) e^{-\sqrt{N} \phi_+(h)} =&\; \gamma_{\| h \|} b(g) + \gamma_{\| h \|} \frac{\gamma_{\| h \|} - 1}{\| h \|^2} \langle g , h \rangle i \phi_- (h) - \frac{\gamma_{\| h \|} - 1}{\| h \|^2} \langle g , h \rangle b^* (h) \\ &- \sqrt{N} \, \gamma_{\|h \|} \frac{\sigma_{\| h \|}}{\| h \|}  \langle g , h \rangle \left( 1 - \frac{\N}{N} \right) + \frac{1}{\sqrt{N}} \, \frac{\sigma_{\| h \|}}{\| h \|} \frac{\gamma_{\| h \|} - 1}{ \| h \|^2}  \langle g , h \rangle a^* (h) a (h) \\ &+ \frac{1}{\sqrt{N}} \, \frac{\sigma_{\| h \|}}{\| h \|} a^* (h) a (g)  \; .
\end{split} 
\end{equation}
\end{proposition}

\begin{remark}
A formula similar to (\ref{eq:prop1-1}) for $e^{\sqrt{N} \phi_+ (h)} b^* (g) e^{-\sqrt{N} \phi_+(h)}$ can be obtained by hermitian conjugation of (\ref{eq:prop1-1}) (and replacing $h$ by $-h$). 
\end{remark}

\begin{proof}
The expressions (\ref{eq:prop1-1}) follows from the commutator expansion
\begin{equation}\label{eq:BHC} e^X \, Y \, e^{-X} = \sum_{j=0}^\infty \frac{1}{j!} \, \ad_X^{(j)} (Y) \end{equation}
combined with the formulas in Lemma \ref{lm:comm_b}. Since the operators $X = \sqrt{N} \phi_+ (h)$ and $Y = b (g)$ are bounded on the truncated Fock space $\mathcal{F}_+^{\leq N}$, it is easy to show the validity of the expansion (\ref{eq:BHC}) for (\ref{eq:prop1-1}) (the difference between $e^X Y e^{-X}$ and $\sum_{j=0}^n    \ad_X^{(j)} (Y) / j!$ converges to  zero in norm, as $n \to \infty$, for every fixed $N \in \mathbb{N}$). 
\end{proof} 

In particular, it follows from (\ref{eq:prop1-1}) that, for $x \in \bR^3$, 
\begin{equation} 
\label{eq:bx-exp} 
\begin{split} 
e^{\sqrt{N} \phi_+ (h)} b_x e^{-\sqrt{N} \phi_+(h)} =&\; \gamma_{\| h \|} b_x + \gamma_{\| h \|} \frac{\gamma_{\| h \|} - 1}{\| h \|^2} h(x) i \phi_- (h) - 
\frac{\gamma_{\| h \|} - 1}{\| h \|^2} h (x) b^* (h) \\ &- \sqrt{N} \, \gamma_{\|h \|} \frac{\sigma_{\| h \|}}{\| h \|}  h (x)  \left( 1 - \frac{\N}{N} \right) + \frac{1}{\sqrt{N}} \, \frac{\sigma_{\| h \|}}{\| h \|} \frac{\gamma_{\| h \|} - 1}{ \| h \|^2}  h (x)  a^* (h) a (h) \\ &+ \frac{1}{\sqrt{N}} \, \frac{\sigma_{\| h \|}}{\| h \|} a^* (h) a_x  \;.
\end{split} 
\end{equation}

We will also need a formula for $e^{\sqrt{N} \phi_+ (h)} a^*_x a_y e^{-\sqrt{N} \phi_+(h)}$. To derive such an expression, we compute 
\[ \begin{split} \frac{d}{ds} e^{s \sqrt{N} \phi_+ (h)} a_x^* a_y e^{-s \sqrt{N} \phi_+ (h)} &= \sqrt{N} e^{s \sqrt{N} \phi_+ (h)} [ \phi_+ (h) , a_x^* a_y ] e^{-s \sqrt{N} \phi_+ (h)}  \\ &= \sqrt{N} \, \overline{h(x)} e^{s \sqrt{N} \phi_+ (h)} b_y e^{-s \sqrt{N} \phi_+ (h)} - \sqrt{N} h (y) e^{s \sqrt{N} \phi_+ (h)} b^*_x e^{-s \sqrt{N} \phi_+ (h)} \; .
\end{split} \]
Using (\ref{eq:bx-exp}) (and its hermitian conjugate) and then integrating over $s \in [0;1]$, we arrive at 
\begin{equation}\label{eq:axay} \begin{split}  e^{s \sqrt{N} \phi_+ (h)} &a_x^* a_y e^{-s \sqrt{N} \phi_+ (h)} \\ = \; &a_x^* a_y + \sqrt{N} \, \frac{\sigma_{\| h \|}}{\| h \|} \left( \overline{h(x)} \, b_y - h(y) \, b_x^* \right) \\ &- N \frac{\sigma_{\| h \|}^2}{\| h \|^2} \, \overline{h(x)} h(y)  \, \left( 1- \frac{\cN_+}{N} \right) + \frac{(\gamma_{\| h \|}-1)}{\| h \|^2}  \left( \overline{h(x)} a^* (h) a_y + h (y) a_x^* a (h) \right) \\ &+ \sqrt{N} \frac{\sigma_{\| h \|}}{\| h \|} \frac{(\gamma_{\| h \|}-1)}{\| h \|^2} \, \overline{h(x)} h(y) \, i \phi_- (h) + \left(\frac{\gamma_{\| h \|}-1}{\| h \|^2} \right)^2  \overline{h(x)} h(y) \, a^* (h) a (h)\; . \end{split} \end{equation}

Integrating (\ref{eq:axay}) against the integral kernel of a self-adjoint operator, we can also get a formula for $e^{\sqrt{N} \phi_+ (h)} d\Gamma (H) e^{-\sqrt{N} \phi_+(h)}$, for a self-adjoint operator $H$. 
\begin{proposition}\label{prop:dG-con}
Let $H : D(H) \to L^2_{\perp \varphi} (\bR^3)$ be self-adjoint, with $D(H) \subset L^2_{\perp \varphi} (\bR^3)$ denoting the domain of $H$. Let $h \in D (H)$. Then  
\begin{equation}\label{eq:prop1-2}
\begin{split} 
e^{\sqrt{N} \phi_+(h)} \rd \Gamma (H) e^{-\sqrt{N} \phi_+(h)} =& \; \rd \Gamma (H) + \sqrt{N} \, \frac{\sigma_{\| h \|}}{\| h \|} i \phi_- (Hh) \\ &- N \frac{\sigma_{\| h \|}^2}{\|h \|^2} \langle h, H h \rangle \left(1 - \frac{\N}{N} \right) + \frac{(\gamma_{\| h \|} -1)}{\| h \|^2} (a^* (h) a (Hh) + a^* (Hh) a (h) )
\\ &+ \sqrt{N} \,  \frac{\sigma_{\| h \|}}{\| h \|} \frac{\gamma_{\| h \|} - 1}{\| h \|^2} \langle h, H h \rangle i \phi_- (h)  + \left( \frac{\gamma_{\| h \|} - 1}{\| h \|^2} \right)^2  \langle h, H h \rangle a^* (h) a (h)  \; . \end{split}\end{equation}
\end{proposition}

\begin{proposition} 
\label{prop:eN}
Let $h \in L^2_{\perp \ph} (\bR^3)$ and denote by $\cN_+$ the number of particles operator on $\cF_{\perp \ph}^{\leq N}$. Then, for every $s \in \bR$, 
\begin{equation}\label{eq:prop2} \begin{split}  
e^{-s \cN_+} b (h) e^{s \cN_+} &= e^s  b(h)  , \\
e^{-s \cN_+} b^* (h) e^{s \cN_+} &= e^{-s} b^* (h) , \\
e^{-s \cN_+} \phi_+ (h) e^{s \cN_+} &= \cosh (s) \phi_+ (h) + \sinh (s) i \phi_- (h) , \\
e^{-s \cN_+} i \phi_- (h) e^{s \cN_+} &= \cosh (s) i \phi_- (h) + \sinh (s) \phi_+ (h) \; . \end{split} \end{equation} 
\end{proposition} 
\begin{proof}
From $[b (h) , \cN_+] = b(h)$ and $[b^* (h), \cN_+] = - b^* (h)$, we easily find that
\[ \begin{split} e^{-s \cN_+} b (h) e^{s \cN_+} &= e^s  b(h), \\ 
e^{-s \cN_+} b^* (h) e^{s \cN_+} &= e^{-s} b^* (h) \; .\end{split} \]
Thus 
\[  \begin{split} e^{-s \cN_+} \phi_+ (h) e^{s \cN_+} &= e^s  b(h) + e^{-s} b^* (h) ,\\ 
e^{-s \cN_+} i\phi_- (h) e^{s \cN_+} &= e^s  b(h) - e^{-s} b^* (h) \; . \end{split} \]
Writing $b(h) = (\phi_+ (h) + i\phi_- (h) )/2$ and $b^* (h) = (\phi_+ (h) - i\phi_- (h))/2$, we arrive at (\ref{eq:prop2}). 
\end{proof}

\begin{proposition}\label{prop:partialt} 
Let $t \mapsto \ph_t$ with $\| \ph_t \|_2 = 1$, independently of $t$. Let $t \mapsto h_t$ be a differentiable map, with values in $L^2_{\perp \ph_t} (\bR^3)$. For $\xi_1, \xi_2 \in \cF_{\perp \ph_t}^{\leq N}$ we find 
\begin{equation}\label{eq:partialt} 
\begin{split} 
\Big\langle \xi_1 , \Big[ \partial_t &e^{\sqrt{N} \phi_+ (h_t)} \Big] e^{-\sqrt{N} \phi_+ (h_t)} \xi_2 \Big\rangle \\ = \; &\sqrt{N} \, \frac{\sigma_{\| h_t \|}}{\| h_t \|}  \langle \xi_1 , \phi_+ ( \partial_t h_t) \xi_2 \rangle - 
\sqrt{N} \frac{\sigma_{\| h_t \|}}{\| h_t \|} \frac{\gamma_{\| h_t \|} -1}{\| h_t \|^2} \text{Im} \langle \partial_t h_t , h_t \rangle  \langle \xi_1, \phi_- (h_t) \xi_2 \rangle 
\\ &- \sqrt{N} \, \frac{\sigma_{\| h_t \|} - \| h_t \|}{\| h_t \|^3} \text{Re} \langle \partial_t h_t , h_t \rangle \langle \xi_1, \phi_+ (h_t) \xi_2 \rangle - i N \frac{\sigma_{\| h_t \|}^2}{\| h_t \|^2} \text{Im} \langle \partial_t h_t , h_t \rangle  \langle \xi_1 , (1 - \cN_+ / N) \xi_2 \rangle \\ &+i  \left(\frac{\gamma_{\| h_t \|} - 1}{\| h_t \|^2} \right)^2 \text{Im} \langle \partial_t h_t , h_t \rangle \langle \xi_1 , a^* (h_t ) a(h_t) \xi_2 \rangle \\ &+ 
 \frac{\gamma_{\| h_t \|} - 1}{\| h_t \|^2} \Big\langle \xi_1, \left[ a^* (h_t) a (\partial_t h_t) - a^* (\partial_t h_t) a (h_t) \right] \xi_2 \Big\rangle \; . \end{split}\end{equation}
\end{proposition}
\begin{proof}
For any two bounded operators $A,B$ we can write
\[ e^A - e^B = \left[ e^A e^{-B} - 1 \right] e^B = \left[ \int_0^1 d\tau \, \frac{d}{d\tau} e^{\tau A} e^{-\tau B} \right] e^B 
= \int_0^1 d\tau \, e^{\tau A} (A-B) e^{(1-\tau) B} \; . \]
Hence, if $t \to A_t$ is an operator-valued functions, differentiable in $t$, we find
\[ e^{A_{t+h}} - e^{A_t} = \int_0^1 d\tau \, e^{\tau A_{t+h}} (A_{t+h} - A_t) e^{(1-\tau )A_t}  \]
Dividing by $h$ and letting $h \to 0$, we find 
\[ \partial_t e^{A_t} = \int_0^1 d\tau \, e^{\tau A_t} \partial_t A_t e^{(1-\tau ) A_t}  \; . \]
In particular,
\[  \left[ \partial_t e^{\sqrt{N} \phi_+ (h_t)} \right] e^{-\sqrt{N} \phi_+ (h_t)} = \sqrt{N} \int_0^1 d\tau \, e^{\tau \sqrt{N} 
\phi_+ (h_t)} \phi_+ ( \partial_t h_t)  e^{-\tau \sqrt{N} \phi_+ (h_t)} \; . \] 
With Prop. \ref{prop:b-con}, we find
\[ \begin{split} \Big[ \partial_t &e^{\sqrt{N} \phi_+ (h_t)} \Big] e^{-\sqrt{N} \phi_+ (h_t)} \\ = \; &\sqrt{N} 
\int_0^1 d\tau \left[ \gamma_{\tau \| h_t \|} \phi_+ ( \partial_t h_t) -2 \gamma_{\tau \| h_t \|} \frac{\gamma_{\tau \| h_t \|} - 1}{\| h_t \|^2} \text{Im} \langle \partial_t h_t , h_t \rangle \phi_- (h_t) \right. \\  &- \frac{\gamma_{\tau \| h_t \|} -1}{\| h_t \|^2} \text{Re} \langle \partial_t h_t , h_t \rangle \phi_+ (h_t) - \frac{\gamma_{\tau \| h_t \|} -1}{\| h_t \|^2} \text{Im} \langle \partial_t h_t , h_t \rangle \phi_- (h_t) \\  
&-2i \sqrt{N} \gamma_{\tau \| h_t \|} \frac{\sigma_{\tau \|h_t \|}}{\| h_t \|} \text{Im}  \langle \partial_t h_t , h_t \rangle (1 - \cN_+ / N)+ \frac{2i}{\sqrt{N}} \frac{\sigma_{\tau \| h_t \|}}{\| h_t \|} \frac{\gamma_{\tau \| h_t \|} - 1}{\| h_t \|^2} \text{Im } \langle \partial_t h_t , h_t \rangle a^* (h_t ) a(h_t)  \\ & \left.  + 
\frac{1}{\sqrt{N}} \frac{\sigma_{\tau \| h_t \|}}{\| h_t \|} \left[ a^* (h_t) a (\partial_t h_t) - a^* (\partial_t h_t) a (h_t) \right] \right] \; . \end{split} \]
Integrating over $\tau$, we arrive at (\ref{eq:partialt}). 
\end{proof} 

\section{Proof of main theorem}

To prove Theorem \ref{thm:main}, we start from (\ref{eq:idea1}), writing 
\[ \begin{split}  \mathbb{E}_{\psi_{N,t}} \, e^{\lambda \left[\sum_{j=1}^N (O^{(j)} - \langle \ph_t , O \ph_t \rangle )\right] }  &=  \left\langle \Omega, \cW^*_N (t;0) e^{\lambda d\Gamma (q_t \wt{O}_t q_t) + \lambda \sqrt{N} \phi_+ (q_t O \ph_t )} \cW_N (t;0) \Omega \right\rangle \; . \end{split} \]

\begin{lemma} 
\label{lm:step1}
There exist constants $C, c > 0$ such that
\begin{equation}\label{eq:step1-1} 
\begin{split} &\left\langle \Omega, \cW^*_N (t;0) e^{\lambda d\Gamma (q_t \wt{O}_t q_t) + \lambda \sqrt{N} \phi_+ (q_t O \ph_t )} \cW_N (t;0) \Omega \right\rangle \\ &\hspace{2cm} \leq e^{C N \| O \|^3 \lambda^3} 
 \left\langle \Omega, \cW^*_N (t;0) e^{\lambda \sqrt{N} \phi_+ (q_t O \ph_t)/2} e^{c \lambda \| O \| \mathcal{N}_+ (t)} e^{\lambda \sqrt{N} \phi_+ (q_t O \ph_t)/2} \cW_N (t;0) \Omega \right\rangle \end{split} 
 \end{equation}
for all $\lambda \leq \| O \|^{-1}$.
 \end{lemma} 

\begin{proof}
For $s \in [0;1]$ and a fixed $\kappa > 0$, we define 
\[ \xi_{s} = e^{(1-s) \lambda \kappa \cN_+ (t) / 2} e^{(1-s) \lambda \sqrt{N} \phi_+ (q_t O \ph_t)/2} e^{s \lambda [ d\Gamma (q_t \wt{O}_t q_t ) + \sqrt{N} \phi_+ (q_t O \ph_t) ]/2} \cW_{N} (t;0) \Omega \; . \]
Note that $\xi_s \in \cF_{\perp \ph_t}^{\leq N}$ for all $s \in [0;1]$. Then, we have 
\[ \| \xi_{0} \|^2 = \left\langle \Omega, \cW_N^* (t;0)  e^{\lambda \sqrt{N} \phi_+ (q_t O \ph_t)/2} e^{\lambda \kappa \cN_+ (t)} e^{\lambda \sqrt{N} \phi_+ (q_t O \ph_t)/2} \cW_N (t;0) \Omega \right\rangle \]
and 
\[ \| \xi_{1} \|^2 = \left\langle \Omega, \cW_N^* (t;0)  e^{\lambda [ d\Gamma (q_t \wt{O}_t q_t ) + \sqrt{N} \phi_+ (q_t O \ph_t) ]} \cW_N (t;0) \Omega \right\rangle \; . \]
To compare $\| \xi_{1} \|^2$ with $\| \xi_{0} \|^2$, we compute the derivative 
\[ \partial_s \| \xi_{s} \|^2 = 2\text{Re } \langle \xi_{s} ; \partial_s \xi_{s} \rangle \; . \]
We have $\partial_s \xi_{s} = \cM_{s} \xi_{s}$, with 
\[ \cM_{s} = \frac{\lambda}{2} e^{(1-s) \lambda \kappa \cN_+ (t) /2} e^{(1-s) \lambda \sqrt{N} \phi_+ (q_t O \ph_t) /2} d\Gamma (q_t \wt{O}_t q_t ) e^{-(1-s) \lambda \sqrt{N} \phi_+ (q_t O \ph_t) /2} e^{-(1-s) \lambda \kappa \cN_+ /2} - \frac{\lambda \kappa}{2} \cN_+ (t) \; . \]
With Proposition \ref{prop:dG-con} we find, defining $h_t = (1-s) \lambda q_t O \ph_t$, 
\[ \begin{split} 
e^{(1-s) \lambda \sqrt{N} \phi_+ (q_t O \ph_t)} &\rd \Gamma (q_t \wt{O}_t q_t ) e^{-(1-s) \lambda \sqrt{N} \phi_+ (q_t O \ph_t)} \\ =& \; \rd \Gamma (q_t \wt{O}_t q_t) - N \frac{\sigma_{\| h_t \|}^2}{\| h_t \|^2} \langle h_t , \wt{O}_t  h_t \rangle \left(1 - \frac{\cN_+ (t)}{N} \right) + \left( \frac{\gamma_{\| h_t \|} - 1}{\| h_t \|^2} \right)^2  \langle h_t , \wt{O}_t  h_t \rangle a^* (h_t) a (h_t) \\ &+ \frac{\gamma_{\| h_t \|} -1}{\| h_t \|^2} (a^* (h_t) a (q_t \wt{O}_t h_t) + a^* (q_t \wt{O}_t h_t) a (h_t) ) \\ &+ \sqrt{N} \,  \frac{\sigma_{\| h_t \|}}{\| h_t \|} \frac{\gamma_{\| h_t \|} - 1}{\| h_t \|^2} \langle h_t , \wt{O}_t h_t \rangle i \phi_- (h_t) + \sqrt{N} \, \frac{\sigma_{\| h_t \|}}{\| h_t \|} i \phi_- (q_t \wt{O}_t h_t) \; . \end{split} \]
With Proposition \ref{prop:eN}, we obtain
\[ \begin{split} 
\frac{2}{\lambda} \, \cM_{s} =&\; \rd \Gamma (q_t \wt{O}_t q_t) - N \frac{\sigma_{\| h_t \|}^2}{\| h_t \|^2} \langle h_t , \wt{O}_t  h_t \rangle \left(1 - \frac{\cN_+ (t)}{N} \right) + \left( \frac{\gamma_{\| h_t \|} - 1}{\| h_t \|^2} \right)^2  \langle h_t , \wt{O}_t  h_t \rangle a^* (h_t) a (h_t) \\ &+ \frac{\gamma_{\| h_t \|} -1}{\| h_t \|^2} (a^* (h_t) a (q_t \wt{O}_t h_t) + a^* (q_t \wt{O}_t h_t) a (h_t) ) \\ &+ \sqrt{N} \,  \frac{\sigma_{\| h_t \|}}{\| h_t \|} \frac{\gamma_{\| h_t \|} - 1}{\| h_t \|^2} \langle h_t , \wt{O}_t h_t \rangle 
\left[ \cosh ((1-s)\lambda \kappa /2) i\phi_- (h_t) + \sinh ((1-s) \lambda \kappa /2) \phi_+ (h_t) \right] \\
&+ \sqrt{N} \, \frac{\sigma_{\| h_t \|}}{\| h_t \|} \left[ \cosh ((1-s) \lambda \kappa /2) i \phi_- (q_t \wt{O}_t h_t) + \sinh ((1-s) \lambda \kappa /2) \phi_+ (q_t \wt{O}_t h_t) \right] \\ &- \kappa \cN_+ (t) \; .\end{split} \]
Using the bounds (\ref{eq:bd-bb*}), (\ref{eq:bd-dG}), and the fact that $\| h_t \| \leq \lambda \| O \| \leq 1$ (from the assumption $\lambda \leq \| O \|^{-1}$), we find 
\[ \begin{split} \frac{2}{\lambda} \text{Re } \langle \xi_{s} , \partial_s \xi_{s} \rangle = \; &\frac{2}{\lambda} \text{Re } \langle \xi_{s} , \cM_{s} \xi_{s} \rangle 
\\ \leq \; &\left[ C \| O \| - \kappa \right] \| \cN_+^{1/2} (t) \xi_{s} \|^2 + C \lambda^2 N \| O \|^3  e^{\lambda \kappa} \| \xi_{s} \|^2  \; .\end{split} \]
Choosing $\kappa = c \| O \|$ (which also implies that $\lambda \kappa \leq c$), we conclude that 
\[ \partial_s \| \xi_{N,s} \|^2 \leq C N \| O \|^3 \lambda^3 \| \xi_{N,s} \|^2 \; . \]
By Gronwall, we obtain (\ref{eq:step1-1}).
\end{proof}

\begin{lemma}
\label{lm:step2}
For a bounded self-adjoint operator $O$ on $L^2 (\bR^3)$ with $\| \Delta O (1-\Delta)^{-1} \|_\text{op} < \infty$, we recall the notation $\vertiii{O}$ from (\ref{eq:Otriple}). Recall also that, for $0 \leq s \leq t$,  $f_{s;t}$ denotes the solution of the equation (\ref{eq:fst}). For given $c > 0$ there exists a constant $C > 0$ such that, with the definition 
\begin{equation}\label{eq:kappas} \kappa_s = c \| O \|_\text{op} \, e^{C (\| v \|_1 + \| v \|_\infty) s} + \frac{\vertiii{O}}{\| v \|_1 + \| v \|_\infty} \left( e^{C ( \|v \|_1 + \| v \|_\infty ) s} - 1 \right). 
\end{equation} 
we have   
\[  \begin{split} &\left\langle \Omega, \cW_N (t;0) e^{\lambda \sqrt{N} \phi_+  (q_t O \ph_t)/2} e^{c \| O \| \cN_+ (t)} e^{\lambda \sqrt{N} \phi_+ (q_t O \ph_t)/2} \cW_N (t;0) \Omega \right\rangle  \\ &\hspace{3cm} \leq 
e^{C N \lambda^3 \vertiii{O}^3 \exp (C (1+ \| v \|_1 + \| v \|_\infty) t)}  \left\langle \Omega, e^{\lambda \sqrt{N} \phi_+ (f_{0;t})/2}  e^{\lambda \kappa_t \cN_+ (0)} e^{\lambda \sqrt{N} \phi_+ (f_{0;t})/2}  \Omega \right\rangle \end{split} \]
for all $\lambda \leq  \vertiii{O}^{-1} e^{-C (\| v \|_\infty + \| v \|_{1} ) t}$.
\end{lemma}

\begin{proof}
For $s \in [0;t]$ and with $\kappa_s$ as in (\ref{eq:kappas}), we define 
\[ \xi_t (s) = e^{\lambda \kappa_s \cN_+ (s) /2} e^{\lambda \sqrt{N} \phi_+ (f_{s;t}) /2} \cW_N (s;0) \Omega \in \cF_{\perp \ph_s}^{\leq N} \]
With $\kappa_0 = c \| O \|$, we observe that 
\[ \| \xi_t (0) \|^2 = \langle \Omega,  e^{\lambda \sqrt{N} \phi_+ (f_{0;t}) /2} e^{c \lambda \| O \| \cN_+ (0)} e^{\lambda \sqrt{N} \phi_+ (f_{0;t}) /2} \Omega \rangle \; .\]
and that 
\[ \| \xi_t (t) \|^2 = \left\langle \Omega,  \cW_N (t;0)^* e^{\lambda \sqrt{N} \phi_+ (q_t O \ph_t) /2} e^{\lambda \kappa_t \cN_+ (t)} e^{\lambda \sqrt{N} \phi_+ (q_t O \ph_t) /2} \cW_N (t;0) \Omega \right\rangle \, .  \]
To compare $\| \xi_t (0) \|^2$ with $\| \xi_t (t) \|^2$, we are going to compute the derivative with respect to $s$. Since the two norms are taken on different spaces, it is convenient to embed first the $s$-dependent space $\cF_{\perp \ph_s}^{\leq N}$ into the full, $s$-independent, Fock space $\cF = \bigoplus_{n \geq 0} L^2 (\bR^{3n})^{\otimes_s n}$. To this end, we observe that 
\[ \begin{split} 
\| \xi_t (s) \|^2 = \; &\left\langle \Omega, \cW_N (s;0)^* e^{\lambda \sqrt{N} \phi_+ (f_{s;t}) /2}  e^{\lambda \kappa_s \cN_+ (s)} e^{\lambda \sqrt{N} \phi_+ (f_{s;t}) /2} \cW_N (s;0) \Omega  \right\rangle_{\cF} \\
= \; &\left\langle \Omega, \cW_N (s;0)^* e^{\lambda \sqrt{N} \phi_+ (f_{s;t}) /2}  e^{\lambda \kappa_s \cN} e^{\lambda \sqrt{N} \phi_+ (f_{s;t}) /2} \cW_N (s;0) \Omega  \right\rangle_{\cF} \end{split} \]
where $\cN$ denotes now the number of particles operator on $\cF$. Hence, we obtain 
\begin{equation}\label{eq:partialJ} \partial_s \| \xi_t (s) \|^2 = -i  \left\langle \xi_t (s) ; \left[ \cJ_{N,t} (s) - \cJ^*_{N,t} (s) \right]  \xi_t (s) \right\rangle \end{equation} 
with the generator (this formula holds if we interpret $\cJ_{N,t} (s)$ as a quadratic form 
on $\cF_{\perp \ph_s}^{\leq N}$)
\begin{equation}\label{eq:gene}
\begin{split}  \cJ_{N,t} (s) = \; &\frac{i\lambda}{2} \dot{\kappa}_s \, \cN_+ (s) + e^{\lambda \kappa_s \, \cN_+ (s) / 2} \left[ i \partial_s e^{\lambda \sqrt{N} \phi_+ (f_{s;t}) /2} \right] e^{-\lambda \sqrt{N} \phi_+ (f_{s;t}) /2} e^{-\lambda \kappa_s \, \cN_+  (s) /2} \\ &+ e^{\lambda \kappa_s \cN_+  (s)/ 2} e^{\lambda \sqrt{N} \phi_+ (f_{s;t}) /2} \cL_N (s) e^{-\lambda \sqrt{N} \phi_+ (f_{s;t}) /2} e^{-\lambda \kappa_s \cN_+ (s) / 2} \; . \end{split} \end{equation} 
Remark that only the antisymmetric part of $\cJ_{N,t} (s)$ contributes to the growth of the norm.

Next, we compute $\cJ_{N,t} (s)$, focusing in particular on its antisymmetric component. 
We recall the definition (\ref{eq:cLN}) of the generator $\cL_N (s)$.
%
We introduce the notation $h_{s;t} = \lambda  f_{s;t} /2 \in L^2_{\perp \ph_s} (\bR^3)$. From (\ref{eq:prop1-2}), we find, on vectors in $\cF_{\perp \ph_s}^{\leq N}$ (since we consider matrix elements on vectors in $\cF_{\perp \ph_s}^{\leq N}$, we can replace the operator $h_H (s) + K_{1,s}$, which does not leave $L^2_{\perp \ph_s} (\bR^3)$ invariant, with its restriction to $L^2_{\perp \ph_s} (\bR^3)$; this is the reason why we can apply Prop. \ref{prop:dG-con})
\[ \begin {split} 
e^{\lambda \sqrt{N} \phi_+ (f_{s;t}) /2} &d\Gamma (h_H (s) + K_{1,s}) e^{-\lambda \sqrt{N} \phi_+ (f_{s;t}) /2} 
\\ = \; &d\Gamma (h_H (s) + K_{1,s}) + \sqrt{N} \, \frac{\sigma_{\| h_{s;t} \|}}{\| h_{s;t} \|} i \phi_- ((h_H (s) + K_{1,s}) h_{s;t}) \\ &- N \frac{\sigma^2_{\| h_{s;t} \|}}{\| h_{s;t} \|^2} \langle h_{s;t} , (h_H (s) + K_{1,s}) h_{s;t} \rangle (1 - \cN_+ (s) / N) \\ &+ \frac{\gamma_{\| h_{s;t} \|} - 1}{ \| h_{s;t} \|^2} \left(a^* (h_{s;t}) a ((h_H (s) + K_{1,s}) h_{s;t}) + a^*((h_H (s) + K_{1,s}) h_{s;t})   a (h_{s;t}) \right) \\ &+ \sqrt{N} \frac{\sigma_{\| h_{s;t} \|}}{\| h_{s;t} \|} \frac{\gamma_{\| h_{s;t} \|} - 1}{\| h_{s;t} \|^2} \langle  h_{s;t} , (h_H (s) + K_{1,s}) h_{s;t} \rangle i \phi_- (h_{s;t}) \\ &+ \left( \frac{\gamma_{\| h_{s;t} \|} -1}{\| h_{s;t} \|^2} \right)^2 \langle  h_{s;t} , (h_H (s) + K_{1,s}) h_{s;t} \rangle a^* (h_{s;t}) a (h_{s;t}) \; .\end{split} \]

With Prop. \ref{prop:eN}, we obtain, again in the sense of forms on $\cF_{\perp \ph_s}^{\leq N}$, 
\[ \begin {split} 
 e^{\lambda \kappa_s \cN_+ (s)/ 2} & e^{\lambda \sqrt{N} \phi_+ (f_{s;t}) /2} d\Gamma (h_H (s) + K_{1,s}) e^{-\lambda \sqrt{N} \phi_+ (f_{s;t}) /2} e^{-\lambda \kappa_s \cN_+ (s) / 2} \\ 
 = \; &d\Gamma (h_H (s) + K_{1,s}) \\ &+ \sqrt{N} \, \frac{\sigma_{\| h_{s;t} \|}}{\| h_{s;t} \|} \left[ \cosh (\lambda \kappa_s /2)  i \phi_- ((h_H (s) + K_{1,s}) h_{s;t}) - \sinh (\lambda \kappa_s /2) \phi_+ ((h_H (s) + K_{1,s}) h_{s;t}) \right]  \\ &- N \frac{\sigma^2_{\| h_{s;t} \|}}{\| h_{s;t} \|^2} \langle h_{s;t} , (h_H (s) + K_{1,s}) h_{s;t} \rangle (1 - \cN_+ / N) \\ &+ \frac{\gamma_{\| h_{s;t} \|} - 1}{ \| h_{s;t} \|^2} \left(a^* (h_{s;t}) a ((h_H (s) + K_{1,s}) h_{s;t}) + a^*((h_H (s) + K_{1,s}) h_{s;t})   a (h_{s;t}) \right) \\ &+ \sqrt{N} \frac{\sigma_{\| h_{s;t} \|}}{\| h_{s;t} \|} \frac{\gamma_{\| h_{s;t} \|} - 1}{\| h_{s;t} \|^2} \langle  h_{s;t} , (h_H (s) + K_{1,s}) h_{s;t} \rangle \left[ \cosh (\lambda \kappa_s /2) i \phi_- (h_{s;t}) - \sinh (\lambda \kappa_s/2) \phi_+ (h_{s;t}) \right] \\ &+ \left( \frac{\gamma_{\| h_{s;t} \|} -1}{\| h_{s;t} \|^2} \right)^2 \langle  h_{s;t} , (h_H (s) + K_{1,s}) h_{s;t} \rangle a^* (h_{s;t}) a (h_{s;t}) \; .\end{split} \] 
Removing symmetric terms (which do not contribute to (\ref{eq:partialJ})) and focussing on terms that are at most quadratic in $\lambda$ (recall that $h_{s;t} = \lambda f_{s;t} /2$), we arrive at 
\begin{equation}
\label{eq:LN-1}
\begin{split}  e^{\lambda \kappa_s \cN_+ (s) / 2} & e^{\lambda \sqrt{N} \phi_+ (f_{s;t}) /2} d\Gamma (h_H (s) + K_{1,s}) e^{-\lambda \sqrt{N} \phi_+ (f_{s;t}) /2} e^{-\lambda \kappa_s \cN_+ (s) / 2} \\ &\hspace{4cm} = \; \frac{i \lambda \sqrt{N}}{2}  \phi_- ((h_H (s) + K_{1,s}) f_{s;t}) + S_1 + T_1 
 \end{split} 
\end{equation}
where $S_1 = S_1^*$ does not contribute to the antisymmetric part of $\cJ_{N,t} (s)$ and \[ \| T_1 \|_\text{op} \leq C N ( \vertiii{O} e^{Ct} + \kappa_s)^3  \lambda^3.\] 
for all $\lambda > 0$ with $\lambda \| O \| \leq 1$ and $\lambda \kappa_s \leq 1$ for all $s \in [0;t]$. Here we used that \[ \| (h_H (s) + K_{1,s} ) f_{s;t} \|  \leq C \vertiii{O} e^{C t} \] 
for all $s \in [0;t]$, $ t > 0$. This follows from the estimate $\| \ph_t \|_{H^4} \leq C e^{C|t|}$, for a constant $C > 0$ depending on $\| \ph \|_{H^4}$ (propagation of high Sobolev norms for the Hartree equation is standard; see \cite{caze}). 

To handle the quadratic off-diagonal term with kernel $K_{2,s}$ in (\ref{eq:cLN}), we apply (\ref{eq:bx-exp}) (and its hermitian conjugate, with $h$ replaced by $-h$, for $b_x^*$, $b_y^*$) and then Prop. \ref{prop:eN}. 
Removing the symmetric part and keeping track only of contributions that are at most quadratic in $\lambda$, we find  
\[ \begin{split} 
e^{\lambda \kappa_s \cN_+ (s) / 2} & e^{\lambda \sqrt{N} \phi_+ (f_{s;t}) /2} \left( \int \left[ K_{2,s} (x;y) b_x b_y  + \overline{K}_{2,s} (x;y) b_x^* b_y^* \right] \, dx dy \right)  e^{-\lambda \sqrt{N} \phi_+ (f_{s;t}) /2} e^{-\lambda \kappa_s \cN_+ (s) / 2} \\ 
=\; &- \lambda \kappa_s \int \left[ K_{2,s} (x;y) b_x b_y - \overline{K}_{2,s} (x;y)  b^*_x b^*_y \right] \, dx dy \\ 
&- \lambda \sqrt{N} \left[ \left(1- \frac{\cN_+  (s) +1/2}{N} \right) b (\overline{K_{2;s} f_{s;t}}) - 
b^* (\overline{K_{2;s} f_{s;t}}) \left(1- \frac{\cN_+ (s)  +1/2}{N} \right) \right] \\ 
&+ \frac{\lambda}{2\sqrt{N}}  \left[  \int dx dy K_{2,s} (x;y) b^* (f_{s;t}) a_x a_y -   \int dx dy \overline{K}_{2,s} (x,y) a^*_y a_x^*  \, b (f_{s;t}) \right] \\ 
&+ \frac{\lambda}{2\sqrt{N}} \left[  \int dx dy \, K_{2,s} (x;y) a^* (f_{s;t}) a_x b_y - 
\int dx dy \, \overline{K}_{2,s} (x;y) b_y^* a_x^*  \, a(f_{s;t}) \right]  \\ &+ S_2  + T_2 
\end{split} \]
where $S_2 = S_2^*$ and $\| T_2 \|_\text{op}  \leq C N  (\vertiii{O} + \kappa_s)^3 \lambda^3$ for all $s \in [0;t]$, if $\lambda \| O \| \leq 1$ and $\lambda \kappa_s \leq 1$ for all $s \in [0;t]$. Thus, we obtain 
\begin{equation}\label{eq:LN-2} \begin{split}  e^{\lambda \kappa_s \cN_+ (s) / 2} & e^{\lambda \sqrt{N} \phi_+ (f_{s;t}) /2} \left( \frac{1}{2} \int \left[ K_{2,s} (x;y) b_x b_y  + \overline{K}_{2,s} (x;y) b_x^* b_y^* \right] \, dx dy \right)  e^{-\lambda \sqrt{N} \phi_+ (f_{s;t}) /2} e^{-\lambda \kappa_s \cN_+  (s) / 2} \\ 
&\hspace{8cm} = - \frac{i \lambda \sqrt{N}}{2} \phi_- (\overline{K_{2;s} f_{s;t}}) + S_2 + T_2 + i R_2 \end{split} \end{equation} 
where $S_2 = S_2^*$, $\| T_2 \|_\text{op} \leq C N   (\vertiii{O}+ \kappa_s)^3 \lambda^3$ and 
\[ \pm R_2 \leq C (\kappa_s \| v \|_\infty + \vertiii{O}) \lambda \cN_+ (s)  \]
for all $s \in [0;t]$ and all $\lambda > 0$ with $\lambda \| O \| \leq 1$ and $\lambda \kappa_s \leq 1$ for all $s \in [0;t]$ (here we used that $\| K_{2,s} \|_\text{op} \leq \| K_{2,s} \|_\text{HS} \leq \| v \|_\infty$ for all $s \in [0;t]$).

Setting $d_s = (v * |\ph_s|^2) + K_{1,s}$ and using Prop. \ref{prop:dG-con} and then Prop. \ref{prop:eN}, 
we obtain  
\[ \begin{split} 
&e^{\lambda \kappa_s \cN_+ (s) / 2}  e^{\lambda \sqrt{N} \phi_+ (f_{s;t}) /2} d\Gamma (d_s) (\cN_+ (s) /N) e^{-\lambda \sqrt{N} \phi_+ (f_{s;t}) /2} e^{-\lambda \kappa_s \cN_+ (s) / 2} \\ = \; & \frac{1}{2\sqrt{N}}  \left[  d\Gamma  (d_s) i \phi_- (h_{s;t}) + i \phi_- (h_{s;t}) d\Gamma (d_s) \right] + \frac{1}{2\sqrt{N}} \left[ i\phi_- (d_s h_{s;t}) \cN_+ + \cN_+ i \phi_- (d_s h_{s;t}) \right] + S_3 + T_3 
\end{split} \]
with $S_3^* = S_3$ and $\| T_3 \|_\text{op} \leq C N  (\vertiii{O} + \kappa_s)^3  \lambda^3$. We conclude that 
\begin{equation}\label{eq:LN-3} e^{\lambda \kappa_s \cN_+ (s) / 2}  e^{\lambda \sqrt{N} \phi_+ (f_{s;t}) /2} d\Gamma (d_s) (\cN_+ (s) /N) e^{-\lambda \sqrt{N} \phi_+ (f_{s;t}) /2} e^{-\lambda \kappa_s \cN_+ (s)  / 2} =  S_3 + T_3 + i R_3 \end{equation} 
where $S_3 = S_3^*$, $\| T_3 \|_\text{op} \leq C N  (\vertiii{O} + \kappa_s)^3  \lambda^3$ and 
\[ \pm R_3 \leq C \vertiii{O}  \lambda \cN_+ (s)  \] 
for all $s \in [0;t]$ and all $\lambda > 0$ with $\lambda \| O \| \leq 1$ and $\lambda \kappa_s \leq 1$ for all $s \in [0;t]$. 

We consider now 
\[ \mathcal{C} = \frac{1}{\sqrt{N}} \int dx dy \, v(x-y) \, \left[ b_x^* a_y^* a_x + a_x^* a_y b_x \right] \]
Conjugating separately $b_x^*$ and $a_y^* a_x$ (or $a_x^* a_y$ and $b_x$ in the second term), we arrive, using (\ref{eq:bx-exp}) (and its hermitian conjugate),  (\ref{eq:axay}) and then Prop. \ref{prop:eN}, at 
\[ \begin{split}  
e^{\lambda \kappa_s \cN_+ (s) / 2} &e^{\sqrt{N} \phi_+ (h_{s;t})} \mathcal{C} e^{-\sqrt{N} \phi_+ (h_{s;t})} e^{- \lambda \kappa_s \cN_+ (s) / 2} 
\\ =\; &\frac{\lambda \kappa_s}{2\sqrt{N}} \int dx dy \, v(x-y) \, \left[ b_x^* a_y^* a_x - a_x^* a_y b_x \right] - \frac{\lambda}{2}  \int  dx dy \, v(x-y)  \left[ f_{s;t} (y) b_x^* b_x - \overline{f_{s;t} (y)} b_x^* b_x  \right] \\
&- \frac{\lambda}{2} \int dx dy \, v(x-y) \left[ f_{s;t} (x) b_x^* b_y^* - \overline{f_{s;t} (x)} b_y b_x  \right] \\ &+ \frac{\lambda}{2}  \int dx dy \, v(x-y) \left[ \overline{f_{s;t} (x)} (1-\cN_+ / N) a_y^* a_x - f_{s;t} (x) a_x ^* a_y (1-\cN_+ /N) \right]  \\
&- \frac{\lambda}{2} \frac{1}{N}  \int dx dy \, v(x-y) \left[ a_x^* a(f_{s;t} ) a_y^* a_x - a_x^* a_y a^* (f_{s;t}) a_x \right] + S_4 + T_4 \end{split} \]
where $S_4^* = S_4$ and $\| T_4 \|_\text{op} \leq C N  ( \vertiii{O}   + \kappa_s)^3 \lambda^3$, for all $s \in [0;t]$ and all $\lambda > 0$ with $\lambda \| O \| \leq 1$ and $\lambda \kappa_s \leq 1$ for all $s \in [0;t]$. We obtain that 
\begin{equation}\label{eq:LN-4} e^{\lambda \kappa_s \cN_+ (s) / 2} e^{\sqrt{N} \phi_+ (h_{s;t})} \mathcal{C} e^{-\sqrt{N} \phi_+ (h_{s;t})} e^{- \lambda \kappa_s \cN_+ (s) / 2}  = S_4 + T_4 + i R_4 \end{equation} 
where $S_4 = S_4^*$, $\| T_4 \|_\text{op} \leq C N ( \vertiii{O}   + \kappa_s)^3   \lambda^3$ and 
\[ \pm R_4 \leq C (\vertiii{O} +  (\| v \|_1 + \| v \|_\infty ) \kappa_s) \lambda \cN_+ (s)  \, .\] 

Finally, we consider the term 
\[ \mathcal{V} = \frac{1}{2N} \int dx dy \, v(x-y) a_x^* a_y^* a_y a_x = \frac{1}{2N} \int dx dy \, v (x-y) a_x^* a_x a_y^* a_y - \frac{v(0)}{2N} \cN_+ (s) \; .\]
Conjugating separately $a_x^* a_x$ and $a_y^* a_y$ (and also the operator $\cN_+ (s)$, using Prop. \ref{prop:dG-con}), we obtain 
\[  \begin{split}  
 e^{\lambda \kappa_s \cN_+ (s) / 2} &e^{\lambda \sqrt{N} \phi_+ (f_{s;t})/2} \mathcal{V} e^{- \lambda \sqrt{N} \phi_+ (f_{s;t})/2} e^{- \lambda \kappa_s \cN_+ (s) / 2} \\ & \hspace{4cm} =
\frac{\lambda}{2\sqrt{N}} \int dx dy \, v(x-y) \, \left[ a_x^* a_x \overline{f_{s;t} (y)} b_y - b_y^* f_{s;t} (y) a_x^* a_x \right] +S_5 + T_5 \end{split} \] 
where $S_5 = S_5^*$ and $\| T_5 \|_\text{op} \leq C N ( \vertiii{O} + \kappa_s)^3  \lambda^3$. Thus 
\begin{equation}\label{eq:LN-5} e^{\lambda \kappa_s \cN_+ (s) / 2} e^{\lambda \sqrt{N} \phi_+ (f_{s;t})/2} \mathcal{V} e^{- \lambda \sqrt{N} \phi_+ (f_{s;t})/2} e^{- \lambda \kappa_s \cN_+ (s) / 2} = S_5 + T_5 +  iR_5 \end{equation} 
with $S_5 = S_5^*$, $\| T_5 \|_\text{op} \leq C N ( \vertiii{O} + \kappa_s)^3  \lambda^3$ and\[ \pm R_5 \leq C \vertiii{O}  \lambda \cN_+ (s) \] 
for all $s \in [0;t]$ and all $\lambda > 0$ with $\lambda \| O \| \leq 1$ and $\lambda \kappa_s \leq 1$ for all $s \in [0;t]$. 

Combining (\ref{eq:LN-1}), (\ref{eq:LN-2}), (\ref{eq:LN-3}), (\ref{eq:LN-4}) and (\ref{eq:LN-5}), we conclude that 
\[ \begin{split} 
e^{\lambda \kappa_s \cN_+ (s) / 2} &e^{\sqrt{N} \phi_+ (h_{s;t})} \cL_N (s) e^{-\sqrt{N} \phi_+ (h_{s;t})} e^{- \lambda \kappa_s \cN_+ (s) / 2} \\ &\hspace{4cm} = \frac{i \lambda \sqrt{N}}{2}  \phi_- ((h_H (s) + K_{1,s} + J K_{2,s}) f_{s;t}) + S +T + i R  \end{split}  \]
where $S^* = S$, $\| T \|_\text{op} \leq C N ( \vertiii{O} e^{Ct} + \kappa_s)^3 \lambda^3$ and \[ \pm R \leq C \lambda (\vertiii{O}  + (\| v \|_\infty + \| v \|_1) \kappa_s) \cN_+ (s)  \]
for all $s \in [0;t]$ and all $\lambda > 0$ with $\lambda \| O \| \leq 1$ and $\lambda \kappa_s \leq 1$ for all $s \in [0;t]$. 

Let us now focus on the second term on the r.h.s. of (\ref{eq:gene}). With Prop. \ref{prop:partialt} we find, in the sense of forms on $\cF_{\perp \ph_s}^{\leq N}$ and keeping track only of contributions that are antisymmetric and at most quadratic in $\lambda$,
\[  \begin{split} e^{\lambda \kappa_s \cN_+ (s) / 2} &\left[ i \partial_s e^{\lambda \sqrt{N} \phi_+ (f_{s;t}) /2} \right] e^{-\lambda \sqrt{N} \phi_+ (f_{s;t}) /2} e^{-\lambda \kappa_s  \cN_+ (s) /2} =  - \frac{i \lambda \sqrt{N}}{2} \,  \phi_- ( i \partial_s f_{s;t}) + \wt{S} + \wt{T}
\end{split}  \]
where $\wt{S} = \wt{S}^*$ and $\| \wt{T} \|_\text{op} \leq C N ( \vertiii{O}+ \kappa_s)^3 \lambda^3$. 

From (\ref{eq:fst}) and (\ref{eq:gene}) we conclude that 
\[ \begin{split}  \pm \frac{1}{i} \left[ \cJ_{N,t} (s) - \cJ_{N,t}^* (s) \right]  \leq \; &C N ( \vertiii{O} e^{Ct} + \kappa_s)^3 \lambda^3 + \lambda \left[ C (\vertiii{O} + ( \| v \|_\infty  + \| v \|_1) \kappa_s) -  \dot{\kappa}_s \right] \cN_+ (s) \end{split} \]
for all $s \in [0;t]$ and all $\lambda > 0$ with $\lambda \| O \| \leq 1$ and $\lambda \kappa_s \leq 1$ for all $s \in [0;t]$. 

With the choice (\ref{eq:kappas}), we find 
\[ \begin{split} 
\pm \frac{1}{i} \left[ \cJ_{N,t} (s) - \cJ_{N,t}^* (s) \right]  \leq \; &C N ( \vertiii{O}e^{Ct} + \kappa_s)^3 \lambda^3  \leq C N \vertiii{O}^3  e^{C  (1+  \| v \|_1 + \| v \|_\infty) t} \lambda^3\end{split} \]
for all $s \in [0;t]$ and all $\lambda \leq C \vertiii{O}^{-1} e^{-C (\| v \|_1 + \| v \|_\infty) t}$. 

Inserting in (\ref{eq:partialJ}) we obtain that 
\[ \left| \partial_s \| \xi_t (t) \|^2 \right| \leq C N \lambda^3 \vertiii{O} e^{C  (1+  \| v \|_1 + \| v \|_\infty) t} \, \| \xi_t (s) \|^2  \; .\]
By Gronwall, we arrive at
\[ \| \xi_t (t) \|^2 \leq e^{C N \lambda^3 \vertiii{O}^3 \exp (C (1+ \| v \|_1 + \| v \|_\infty) t)} \, \| \xi_t (0) \|^2 \]
for all $s \in [0;t]$ and all $\lambda \leq \vertiii{O}^{-1} e^{-C (\| v \|_1 + \| v \|_\infty) t}$. 
\end{proof}

\begin{lemma} \label{lm:step3}
Let $\kappa_t$ be defined as in (\ref{eq:kappas}). Then there exists a constant $C > 0$ such that 
\begin{equation}\label{eq:step3} \left\langle \Omega, e^{\lambda \sqrt{N} \phi_+ (f_{0;t})/2}  e^{\lambda \kappa_t \cN_+ (0)} e^{\lambda \sqrt{N} \phi_+ (f_{0;t})/2}  \Omega \right\rangle \leq e^{\lambda^2 N \| f_{0;t} \|^2/2} e^{C N \lambda^3 \vertiii{O}^3 \exp (C (\| v \|_1 + \| v \|_\infty) t)} \end{equation}
for all $\lambda \leq \vertiii{O}^{-1} e^{-C ( \| v \|_\infty + \| v \|_1) t}$ and all $t > 0$. 
\end{lemma}

{\it Remark:} The lemma could be extended to bound the expectation on the l.h.s. of (\ref{eq:step3}) 
for a larger class of states, including quasi-free states, rather than only in the vacuum. This 
would allow us to consider more general initial data in Theorem \ref{thm:main}. To keep the focus on the main novelty of our paper (the possibility of proving a large deviation principle for many-body quantum dynamics), we restricted our attention on the simplest case of factorized initial data (leading to the vacuum in (\ref{eq:step3}).

\begin{proof}
For $s \in [0;1]$ and setting $h_t = \lambda f_{0;t}/2 \in L^2_{\perp \ph} (\bR^3)$, we define  
\[ \xi_s = e^{\lambda \kappa_t \cN_+ (0) /2} e^{s \sqrt{N} \phi_+ (h_t)} e^{(1-s) \sqrt{N} b^* (h_t)} e^{(1-s) \sqrt{N} b(h_t)} e^{(1-s)^2 N \| h_t \|^2 /2} \Omega \; .\]
Then 
\[ \| \xi_1 \|^2 =  \left\langle \Omega, e^{\lambda \sqrt{N} \phi_+ (f_{0;t})/2}  e^{\lambda \kappa_t \cN_+ (0)} e^{\lambda \sqrt{N} \phi_+ (f_{0;t})/2}  \Omega \right\rangle \]
is the quantity we want to estimate, while
\begin{equation}\label{eq:step3-1} \| \xi_0 \|^2 = e^{N \| h_t \|^2} \langle e^{\sqrt{N} b^* (h_t)} \Omega , e^{\lambda \kappa_t \cN_+ (0)} e^{\sqrt{N} b^* (h_t)} \Omega \rangle \end{equation}
is going to give the bound on the r.h.s. of (\ref{eq:step3}). 

To compare $\| \xi_1 \|^2$ with $\| \xi_0 \|^2$, we compute the derivative
\begin{equation}\label{eq:norm-gr} \partial_s \| \xi_s \|^2 = 2 \text{Re } \langle \xi_s , \cG_s \xi_s \rangle \end{equation}
where 
\[ \begin{split}  \cG_s = \; &- (1-s) N \| h_t \|^2 \\ &+ \sqrt{N}  e^{\lambda \kappa_t \cN_+ (0) /2} e^{s \sqrt{N} \phi_+ (h_t)} \\ &\hspace{1.5cm} \times  \left[ \phi_+ (h_t) -  b^* (h_t) -  e^{(1-s) \sqrt{N}  b^* (h_t)} b (h_t) e^{-(1-s) \sqrt{N} b^* (h_t)} \right] e^{-s\sqrt{N} \phi_+ (h_t)} e^{-\lambda \kappa_t \cN_+ (0) /2} \end{split} \]
is defined so that $\partial_s \xi_s = \cG_s \xi_s$. With the commutation relations (\ref{eq:comm1})-(\ref{eq:comm3}), we find the identity 
\[ \begin{split} e^{(1-s) \sqrt{N} b^* (h_t)} &b(h_t) e^{-(1-s) \sqrt{N} b^* (h_t)} \\ &= b(h_t) - \sqrt{N} \| h_t \|^2 (1-s) \left(1- \frac{\cN_+ (0)}{N} \right) - \| h_t \|^2 (1-s)^2 b^* (h_t) + \frac{(1-s)}{\sqrt{N}} a^* (h_t) a(h_t) \; .\end{split} \] 
Thus 
\[ \begin{split} \cG_s =  - e^{\lambda \kappa_t \cN_+ (0) /2} e^{s \sqrt{N} \phi_+ (h_t)} &\left[ (1-s) \| h_t \|^2 \cN_+ (0) + (1-s) a^* (h_t) a(h_t) - \sqrt{N} \| h_t \|^2 (1-s)^2 b^* (h_t) \right] \\ &\hspace{7cm} \times e^{-s \sqrt{N} \phi_+ (h_t)} e^{-\lambda \kappa_t \cN_+ (0) /2} \; . \end{split} \]
With Prop. \ref{prop:b-con} and Prop. \ref{prop:dG-con} we obtain 
\[ \cG_s = - (1-s) \| h_t \|^2 \cN_+ (0) - (1-s) a^* (h_t) a(h_t) + T \]
where (using the definition (\ref{eq:kappas}) of $\kappa_t$) 
\[ \| T \| \leq C N  \lambda^3 \vertiii{O}^3 e^{C (\| v \|_1 + \| v \|_\infty) t}, \]
for all $\lambda \leq \vertiii{O}^{-1} e^{-C ( \| v \|_\infty + \| v \|_1) t}$ (this guarantees that $\lambda \kappa_t \leq 1$ and $\lambda \| O \| \leq 1$). 
From (\ref{eq:norm-gr}), we obtain 
\[  \partial_s \| \xi_s \|^2 \leq C N \lambda^3 \vertiii{O}^3  e^{C (\| v \|_1 + \| v \|_\infty) t} \| \xi_s \|^2 \]
and thus that 
\[ \| \xi_1 \|^2 \leq e^{CN \lambda^3 \vertiii{O}^3  \exp (C (\| v \|_1 + \| v \|_\infty) t)}   \| \xi_0 \|^2 \]
for all $\lambda \leq \vertiii{O}^{-1} e^{-C ( \| v \|_\infty + \| v \|_1) t}$.

It remains to compute
\[ \begin{split} \| \xi_0 \|^2 &=  e^{N \| h_t \|^2}  \langle e^{\sqrt{N} b^* (h_t)} \Omega , e^{\lambda \kappa_t \cN_+ (0) }e^{\sqrt{N} b^* (h_t)} \Omega \rangle = e^{N \| h_t \|^2} \sum_{n=0}^N \frac{N^n}{(n!)^2} e^{\lambda \kappa_t n} \| b^* (h_t)^n \Omega \|^2 \; .\end{split} \]
Notice that 
\[ \begin{split} \| b^* &(h_t)^n \Omega \|^2 \\ &= \left\| a^* (h_t) (1- \cN_+ (0) /N)^{1/2} a^* (h_t) (1-\cN_+ (0)/N)^{1/2} \dots a^* (h_t) (1-\cN_+ (0)  /N)^{1/2} \Omega \right\|^2 \\ &= \left\| a^* (h_t)^n (1- (\cN_+ (0) + n -1)/N)^{1/2} (1- (\cN_+ (0) +n -2)/N)^{1/2} \dots (1- \cN_+ (0) / N)^{1/2} \Omega \right\|^2 \\ &= \frac{(N - (n-1)) \dots (N-1)}{N^{(n-1)}} \| a^* (h_t)^n \Omega \|^2 =  \frac{(N-1)!}{N^{(n-1)} (N-n)!} n!  \| h_t \|^{2n} \; . \end{split} \]
Therefore, recalling that $h_t = \lambda f_{0;t}/2$
\[ \begin{split} \| \xi_0 \|^2 &= e^{N \| h_t \|^2} \sum_{n=0}^N {N \choose n} \| h_t \|^{2n}  e^{\lambda \kappa_t n}  \\ &= e^{N \| h_t \|^2} \, \left(1 + \| h_t \|^2 e^{\lambda \kappa_t} \right)^N \\ & \leq e^{N \| h_t \|^2 (1 + e^{\lambda \kappa_t})}  \leq e^{N \lambda^2  \| f_{0;t} \|^2 / 2} e^{C N \lambda^3 \vertiii{O}^3 \exp (C (\| v \|_\infty + \| v \|_1) t)} \end{split} \]
for all $\lambda \leq  \vertiii{O}^{-1} e^{-C ( \| v \|_\infty + \| v \|_1) t}$. We conclude that 
\[ 
 \left\langle \Omega, e^{\lambda \sqrt{N} \phi_+ (f_{0;t})/2}  e^{\lambda \kappa_t \cN_+ (0)} e^{\lambda \sqrt{N} \phi_+ (f_{0;t})/2}  \Omega \right\rangle \leq e^{N \lambda^2  \| f_{0;t} \|^2 / 2} e^{C N \lambda^3 \vertiii{O}^3 \exp (C (\| v \|_\infty + \| v \|_1) t)}  \]
 for all $\lambda \leq  \vertiii{O}^{-1} e^{-C ( \| v \|_\infty + \| v \|_1) t}$. 
\end{proof}

\begin{proof}[Proof of Theorem \ref{thm:main}]
Combining Lemma \ref{lm:step1}, Lemma \ref{lm:step2} and Lemma \ref{lm:step3}, we arrive at
\[ \left\langle \Omega, \cW^*_N (t;0) e^{\lambda d\Gamma (q_t \wt{O}_t q_t) + \lambda \sqrt{N} \phi_+ (q_t O \ph_t )} \cW_N (t;0) \Omega \right\rangle \leq e^{N \lambda^2  \| f_{0;t} \|^2 / 2}  \, e^{C N \lambda^3 \vertiii{O}^3 \exp (C (1+ \| v \|_1 + \| v \|_\infty) t)} \; . \]
Therefore 
\[ \begin{split}
\frac{1}{N} \log \mathbb{E}_{\psi_{N,t}} \, e^{\lambda \left[\sum_{j=1}^N (O^{(j)} - \langle \ph_t , O \ph_t \rangle ) \right]}  &= \frac{1}{N} \log \, \left\langle \Omega, \cW^*_N (t;0) e^{\lambda d\Gamma (q_t \wt{O}_t q_t) + \lambda \sqrt{N} \phi_+ (q_t O \ph_t )} \cW_N (t;0) \Omega \right\rangle \\  &\leq \frac{\lambda^2}{2}  \| f_{0;t} \|^2  + C \lambda^3 \vertiii{O}^3 \exp (C (1+ \| v \|_1 + \| v \|_\infty) t) \end{split}  \]
for all $\lambda \leq  \vertiii{O}^{-1} e^{-C ( \| v \|_\infty + \| v \|_1) t}$.

 \end{proof}

\paragraph{Acknowledgements.} 
The authors gratefully acknowledge G\'erard Ben Arous for suggesting this kind of result. K.L.K. was partially supported by NSF CAREER Award DMS-125479 and a Simons Sabbatical Fellowship. S.R. acknowledges funding from the European Union's Horizon 2020 research and innovation programme under the Marie Sk\l{}odowska-Curie Grant Agreement No. 754411. B. S. gratefully acknowledges partial support from the NCCR SwissMAP, from the Swiss National Science Foundation through the Grant ``Dynamical and energetic properties of Bose-Einstein condensates'' and from the European Research Council through the ERC-AdG CLaQS.

\end{document}